\documentclass[runningheads]{llncs}

\usepackage{amsmath,amssymb,graphicx,color, color, hyperref} 

\usepackage{natbib}			

\newtheorem{assumption}{Assumption}

\renewenvironment{proof}[1]{{\em #1}\hspace{1em}}{\hfill$\Box$\medskip\medskip}

\newcommand{\cB}{{\mathcal{B}}}

\def\eps{{\epsilon}}

\def\H{{\rm H}}

\def\H{\cal H}

\def\M{{\cal M}}

\def\E{\mbox{E}}

\def\N{{\mathcal{N}}}

\def\Rev{{\mbox{\sc Revenue}}}
\def\rev1{\mbox{\sc rev}_1}

\def\rs{{r^{\star}_s}}
\def\rL{{r^{\star}_L}}
\def\rH{{r^{\star}_H}}

\begin{document}

\authorrunning{Kanoria  Nazerzadeh}

\author{Yash Kanoria \ \ \ \ \ \ \ \ \ \ \ \ \ \ \ \ \ \ \  Hamid Nazerzadeh}


\institute{Columbia University \ \ \ \ \ \  University of Southern California}

\title{Dynamic Reserve Prices for Repeated Auctions: \\ Learning from Bids\thanks{This paper appeared as a one page abstract in the proceedings of \emph{Web and Internet Economics $($WINE $)$}, 2014, pp. 232--232. A subsequent paper by the same authors, ``Incentive-Compatible Learning of Reserve prices for Repeated Auctions,'' to appear in \emph{Operations Research}, features a different model (with bidders drawing i.i.d. valuations across rounds) for a similar problem.}}

%
%
%
%

\maketitle

\medskip
\makebox[\linewidth]{July 2014\phantom{xxxxxxx}}

\begin{abstract}
A large fraction of online advertisement is sold via repeated second price auctions. In these auctions, the reserve price is the main tool for the auctioneer to boost revenues. In this work, we investigate the following question: Can changing the reserve prices based on the previous bids improve the revenue of the auction, taking into account the long-term incentives and strategic behavior of the bidders? We show that if the distribution of the valuations is known and satisfies the standard regularity assumptions, then the optimal mechanism has a constant reserve. However, when there is uncertainty in the distribution of the valuations, previous bids can be used to learn the distribution of the valuations and to update the reserve price. We present a simple, approximately incentive-compatible, and asymptotically optimal dynamic reserve mechanism that can significantly improve the revenue over the best static reserve.
\end{abstract}

\section{Introduction} \label{sec:intro}
Advertising is the main component of monetization strategies of most Internet companies. A large fraction of online advertisements are sold via advertisement exchanges platforms such as Google's Doubleclick (Adx) and Yahoo!'s Right Media.\footnote{Other examples of major ad exchanges include Rubicon, AppNexus, and OpenX.}  Using these platforms, online publishers such as the New York Times and the Wall Street Journal sell the advertisement space on their webpages to advertisers. The advertisement space is allocated using auctions where advertisers bid in real time for a chance to show their ads to the users.
Every day, tens of billions of online ads are sold via these exchanges~\citep{Muthu10,McAfee11,BalseiroFMM11,CelisLMN14}.

The second-price auction is the dominant mechanism used by the advertisement exchanges.
Among the reasons for such prevalence are the simplicity of the second-price auction and the fact that it incentivizes the advertisers to be truthful. The second price auction maximizes the social welfare (i.e., the value created in the system) by allocating the item to the highest bidder.

In order to maximize the revenue in a second price auction, the auctioneer can set a reserve price and not make any allocations when the bids are low. In fact, under symmetry and regularity assumptions (see
Section~\ref{sec:model}),  the second-price auction with an appropriately chosen reserve price is optimal and maximizes the revenue among all selling mechanism~\citep{Myerson81,RileySamuelson81}.

However, in order to set the reserve price effectively, the auctioneer requires information about distribution of the valuations of the bidders. A natural idea, which is widely used in practice, is to construct these distributions using the history of the bids. This approach, though intuitive, raises a major concern with regards to long-term (dynamic) incentives of the advertisers. Because the bid of an advertiser may determine the price he or she pays in future auctions, this approach may result in the advertisers shading their bids and ultimately in a loss of revenue for the auctioneer.

To understand the effects of changing reserve prices based on the previous bids, we  study a setting where the auctioneer sells impressions (advertisements space) via repeated second price auctions. We demonstrate that the long-term incentives of advertisers plays an important role in the performance of these repeated auctions by showing that under standard symmetry and regularity assumptions (i.e., when the valuations of are drawn independently and identically from a regular distribution), the optimal mechanism is running a second price auction with a {\em constant} reserve and changing the reserve prices over time is not beneficial.
However, when there is {\em uncertainty} in the distribution of the valuations, we show that there can be substantial benefit in learning the reserve prices using the previous bids.



More precisely, we consider an auctioneer selling multiple copies of an item  sequentially. The item is either a high type or  a low type. The type determines the distribution of the valuations of the bidders. The type of the item is not a-priori known to the auctioneer. Broadly, we show the following: when there is competition between bidders and the valuation distributions for the two types are sufficiently different from each other, there is a simple dynamic reserve mechanism that can effectively ``learn" the type of the item, and thereafter choose the optimal reserve for that type.\footnote{On the other hand, when the valuation distributions for the two types are close to each other, the improvement from changing the reserve is insignificant.} As a consequence, the dynamic reserve mechanism does much better than the best fixed reserve mechanism, and in fact, achieves near optimal revenue, while retaining (approximate) incentive compatibility.\footnote{Approximate incentive compatibility implies that the agent behave truthfully if the gain from deviation is small; see Section~\ref{sec:model}.}

%
To this end, we propose a simple mechanism called the {\em threshold} mechanism. In each round, the mechanism implements a second price auction with reserve. The reserve price starts at some value, and stays there until there is a bid exceeding a pre-decided threshold, after which the reserve rises (permanently) to a higher value.


We compare the revenue of our mechanism with two benchmarks. Our baseline is the static second price auction with the optimal constant reserve. Our upper-bound benchmark is the optimal mechanism that knows the type of the impressions (e.g., high or low) in advance. These two benchmarks are typically well separated. We show that the threshold mechanism is near optimal and obtains revenue close to the upper-bound benchmark. In addition, we present numerical illustrations of our results that show up to $23\%$ increase in revenue by our mechanism compared with the static second price auctions. These examples demonstrate the effectiveness of dynamic reserve prices under fairly broad assumptions.




\subsection{Related Work}
In this section, we briefly discuss the closest work to ours in the literature along different dimensions starting with the application in online advertising. 

\cite{OstrovskyS09} conducted a large-scale field experiment at Yahoo! and showed that  choosing reserve prices, guided by the theory of optimal auctions, can significantly increase the revenue of sponsored search auctions. To mitigate the aforementioned incentive concerns, they dropped the highest bid from each auction when estimating the distribution of the valuations. However, they do not formally discuss the consequence of this approach. 

Another common solution offered to mitigate the incentive constraints is to bundle different types of impressions (or keywords) together so that the bid of each advertiser would have small impact on the aggregate distribution learned from the history of bids. However, this approach may lead to significant estimation errors and setting a sub-optimal reserve.

To the extent of our knowledge, ours  is the first work that rigorously studies the long-term and dynamic incentive issues in repeated auctions with dynamic reserves. 

\cite{IyerJS11} and \cite{BalseiroBW13} demonstrate the importance of setting reserve prices in dynamic setting in environments where agents are uncertain about their own valuations, and respectively, are budget-constrained. We discuss the methodology of these papers in more details at the end of Section~\ref{sec:DIC}. \cite{McAfeeMR89,McAfeeV92} determine reserve prices in common value settings.

Our work is closely related to the literature on behavior-based pricing strategies where the seller changes the prices for a buyer (or a segment of the buyers) based on her previous behavior; for instance, increasing the price after a purchase or reducing the price in the case of no-purchase; see~\cite{FundenbergV07,Esteves09} for surveys.

The common insight from the literature is that the optimal pricing strategy is to commit to a single price over the length of the horizon~\citep{Stokey79,Salant89,HartT88}. In fact, when customers anticipate reduction in the future prices, dynamic pricing may hurt the seller's revenue~\citep{Taylor04,Villas-Boas04}. Similar insights are obtained in environments where the goal is to sell a fixed initial inventory of products to unit-demand buyers who arrive over time~\citep{AvivP08,DasuT10,AvivWZ13,CorreaMT13}.

There has been renewed interest in behavior-based pricing strategies, mainly motivated by the development in e-commerce technologies that enables online retailers and other Internet companies to determine the price for the buyer based on her previous purchases. \cite{AcuistiV05} show that when a sufficient proportion of customers are myopic or when the valuations of customers increases (by providing enhanced services) dynamic pricing may increase the revenue. Another setting where dynamic pricing could boost the revenue is when the seller is more patient than the buyer and discounts his utility over time at a lower rate than the buyer~\citep{BikhchandaniM12,AminRS13}. See~\cite{Taylor04,ConitzerTW12} for privacy issues and anonymization approaches in this context. \\
In contrast with these works, our focus is on auction environments and we study the role of competition among strategic bidders.


The problem of learning the distribution of the valuation and optimal pricing also have been studied in the context of revenue management and pricing for markets where each (infinitesimal) buyer does not have an effect on the future prices and demand curve can be learned with optimal regret~\citep{BesbesZ09,BesbesZ12,HarrisonKZ12,BoerZ13,Segal03}.
In this work, we consider a setting where the goal is to learn the optimal reserve price with strategic and forward looking buyers, with multi-unit demand, where the action of each buyer can change the prices in the future.





\paragraph{Organization}
The remaining of the paper is organized as follows. In Section~\ref{sec:model}, we formally present the model followed by the description of the threshold mechanisms in Section~\ref{sec:threshold}. We show that the mechanism is dynamic incentive compatible in Section~\ref{sec:DIC}. In Sections~\ref{sec:general}, we present an extension of the threshold mechanism.

\section{Model and Preliminaries} \label{sec:model}
%

A seller auctions off $T>1$ items to $n\ge 1$ agents in $T$ rounds of second price auctions, numbered $t=1,2, \ldots, T$. The items are of type high or low denoted by $s \in \{L, H\}$, where informally we think of an item of type $H$ as being more valuable than an item of type $L$. The items are \emph{all of the same type}. The type is $s$ with probability $p_s$.

The valuation of agent $i \in \{1,\ldots, n\}$ for an item of type $s$, denoted by $v_i$, is drawn independently and identically from distribution $F_s$, i.e., the valuations are i.i.d. conditioned on $s$.
Note that agents' valuations  are identical in each round (if they participate, see below). In Section~\ref{sec:transient} we consider an extension of our model where the valuations of the agents may change over time. 

%

Each agent participates in each auction with probability $\alpha_i$ exogenously and independently across rounds and agents. One can think of $\alpha_i$'s as throttling probabilities~\footnote{Due to budget and bandwidth constraints or other considerations, online advertising platforms often randomly select a subset of bidders, from all eligible advertisers, to participate in the auction. This process is referred to as throttling~\citep{GoelMNS10,Charlesetal13}.} or matching probabilities to a specific user demographic~\citep{CelisLMN14}. Let $X_{it}$ be the indicator random variable corresponding to the participation of agent $i$ in auction at time $t$. Note that $\alpha_i = \E[X_{it}]$. We denote the realization of $X_{it}$ by $x_{it}$. Agent $i$ learns $x_{it}$ at the beginning of round $t$.
In particular, our (incentive compatibility) results hold  in the special case when all the agents participate in all the auctions, i.e., $\alpha_i=1$ for all $i$. Participation probabilities allow us to model environments where a small number of bidders participate in auctions (cf.~\cite{CelisLMN14}).
%

\paragraph{Information Structure}
We assume $T$, $p_L$, and $p_H$ to be common knowledge. We also assume that the type of the item $s$ is common knowledge among the agents but unknown to the auctioneer, who only knows $p_s$. This assumption is motivated in part by the application where sometimes advertisers may have more information about the value of a user or an impression than the publisher. Also, it corresponds to a stronger requirement for the incentive compatibility of the mechanism; hence, our results remain valid if the agents have the same information as the seller about the type of the item. Similarly, we assume that $\alpha_i$'s are common knowledge among the agents and the auctioneer. Our mechanism remains incentive compatible, as defined below, if the agents have incomplete information about the $\alpha_i$'s.
At the beginning, each agent $i$ knows his own valuation, $v_i$, but not the other agents' valuations (but agents may make inferences about the valuations of the other agents over time).

%

\medskip
Let us now consider the seller's problem. The seller aims to maximize her expected revenue via a repeated second price auction.


\paragraph{\bf A ``generic" dynamic second price mechanism} 
At time $0$, the auctioneer announces the reserve price function $\Omega: \H \rightarrow \mathbb{R}^+$ that maps the history observed by the mechanism to a reserve price. The history observed by the mechanism up to time $\tau$, denoted by $H_{\Omega,\tau}\in \H$, consists of, for each round $t< \tau$, the reserve price, the agents participating in round $t$ and their bids, and the allocation and payments at that round. More precisely, $$H_{\Omega,\tau} =~\langle (r_1,x_{1},b_{1},q_{1},p_{1}), \cdots,(r_{\tau-1},x_{\tau-1},b_{\tau-1},q_{\tau-1},p_{\tau-1})\rangle$$ where
\begin{itemize}
\item $r_t$ is the reserve price at time $t$.

\item $x_t=~\langle x_{1t},\cdots,x_{nt}\rangle$. Recall that $x_{it}$ is equal to $1$ if agent $i$ participates in the auction for item $t$. 

\item $b_t=~\langle b_{1t},\cdots,b_{nt}\rangle$ where $b_{it}$ denotes the bid of agent $i$ at time $t$. We assign $b_{it} = \phi$ if $x_{it}=0$, i.e., if agent $i$ does not participate in round $t$.

\item $q_t$ corresponds to the allocation vector. Since the items are allocated via the second price auction with reserve $r_t$, if all the bids are smaller than $r_t$, the item is not allocated. Otherwise, the item is allocated uniformly at random to an agent $i^{\star}\in\arg\max_{i}\{b_{it}\}$ and we have $q_{i^{\star}}=1$. For all the agents that did not receive the item, $q_{it}$ is equal to $0$.

\item $p_t$ is the vector of payments. If $q_{it}=0$, then $p_{it}=0$
and if $q_{it}=1$, then $p_{it}=\max\left\{\max_{j\neq i}\left\{b_{jt}\right\},r_t\right\}$.
\end{itemize}
Note that in our notation, $\Omega$ includes a reserve price function for each $t = 1, 2, \ldots, T$. The length of the history $\H$ implicitly specifies the round for which the reserve is to be computed.

An important special case is a {\em static mechanism} where the reserve is not a function of the previous bids and allocations.

We can now define the seller problem more formally. The seller chooses a reserve price function $\Omega$ that maximizes the expected revenue, which is equal to $\E\left[\sum_{t=1}^{T} \sum_{i=1}^n p_{it}\right]$, when the buyers play an equilibrium with respect to the choice of $\Omega$. In order to define the utility of the agents, let $H_{ik}$ denote the history observed by agent $i$ up to time $t$ including  the allocation and payments of (only) agent $i$. Namely, $$H_{ik}=~\langle (r_1,x_{i,1},b_{i,1},q_{i,1},p_{i,1}), \cdots,(r_{t-1},x_{i,t-1},b_{i,t-1},q_{i,t-1},p_{i,t-1})\rangle.$$
Next, we state precise definitions of a bidding strategy and a best response.

\begin{definition}[Bidding Strategy]
\emph{Bidding strategy $B_i: \mathbb{R} \times {\H}_i \times \mathbb{R} \rightarrow \mathbb{R}$ 
of agent $i$ maps the valuation of the agent $v_i$, history $H_{it}$, and the reserve $r_t$ at time $t$ to a bid $b_{it}=B_i(v_i,H_{it},r_t)$. Here ${\H}_{i}$ is the set of possible histories observed by agent $i$.}
\end{definition}

\begin{definition}[Best-Response] \label{def:best-response}
\emph{Given strategy profile $<B_1, B_2, \cdots, B_n>$, $B_i$ is a best-response strategy to the strategy of other agents $B_{-i}$, if, for all $s$ and $v_i$ in the support of $F_s$, it maximizes the expected utility of agent $i$, $$U_i(v_i, s, B_{i}, B_{-i})=\E\left[\sum_{t=1}^{T} v_i q_{it} - p_{it}\right],$$ where the expectation is over the valuations of other agents, the participation variables $x_{jt}$'s, and any randomization in bidding strategies.
Strategy $B_i$ is an $\epsilon$-best-response if, for all $v_i$ in the support,
$$U_i(v_i, s, B_{i}, B_{-i}) \geq U_i(v_i, s, \textup{BR}(v_i,s, B_{-i}), B_{-i}) + T \alpha_i \eps\, ,$$
where $\textup{BR}$ denotes a best response.}
\end{definition}

A mechanism is incentive compatible if, for each agent $i$, the truthful strategy is a best-response to the other agents being truthful. In this paper, we consider the notion of approximate incentive compatibility that implies that an agent does not deviate from the truthful strategy when the benefit from such deviation is insignificant. This notion is appealing when characterizing, or computing, the best response strategy is challenging and has been studied for static games (cf.~\cite{DaskalakisMP09,ChienS11,KearnsM02,FederNS07,HemonRS08}) as well as dynamic games, such as ours, where finding the best response strategy of an agent corresponds to solving a complicated (stochastic) dynamic program~\citep{IyerJS11,BalseiroBW13,GummadiKP13,NazerzadehSV13}.


\begin{definition}[Approximate Incentive Compatibility] \label{def:IC}
{\em A mechanism is $\epsilon$-incentive compatible if the truthful strategy of agent $i$ is an $\epsilon$-best-response to the truthful strategy of other agents for all $s$ and all $v_i$ in the support of $F_s$.}
\end{definition}

Note that $\alpha_i T$ is the expected number of rounds in which agent $i$ participates. Therefore, under an $\eps$-incentive compatible mechanism, on average the agent loses at most $\eps$ in utility, relative to playing a best response, per-round of participation.

We now define a stronger notion of incentive compatibility. In Section~\ref{sec:DIC}, we provide conditions under which our proposed mechanism satisfies these stronger notion.
By a \emph{realization}, denoted by  $(v_i,x_i^T)_{i=1}^n$, we refer to a valuation vector $(v_1, v_2, \ldots, v_n)$ along with a participation vector $(x_{1}^{T}, x_{2}^T, \ldots, x_{n}^T)$.

\begin{definition}[Dynamic Incentive Compatibility] \label{def:DIC}
{\em We call a realization $\eps$-good with respect to a mechanism, if truthfulness, for each agent $i$ and in {\em each round} $t\in\{1,2, \ldots, T\}$, remains an (additive) $\eps\alpha_i(T-t)$-best-response to the truthful strategy of the other agents.
We say that a mechanism is \emph{$(\delta,\epsilon)$-dynamic-incentive-compatible} if
 the probability of the realization being $\eps$-good with respect to the mechanism is at least $1-\delta$.
 }
\end{definition}

Thus, in a $(\delta,\epsilon)$-dynamic-incentive-compatible mechanism, assuming truthful bidding, with probability at least $1-\delta$ the realization satisfies the following property: for  each bidder and each round that the bidder participates, the average cost of truthful bidding is at most $\eps$ for each future round that he may participate in, relative to a best response. The above definition extends the notion of (exact) interim dynamic incentive compatibility~\citep{BergemannV10} which implies that the agents will not deviate from the truthful strategy even as they obtain more information over time.

\subsection*{Benchmarks}
In the next section, we propose a simple approximately incentive compatible mechanism for the setting described above.
We compare our proposed mechanism with two benchmarks that provide a lower-bound and an upper-bound on the revenue of the best dynamic second price mechanism.

The lower-bound mechanism, which we refer to as the {\em static mechanism}, at each step, implements a second price auction with a constant reserve $r_0$. The reserve is chosen at time $0$ before the mechanism observes any of the bids and does not change over time.\footnote{
Since the valuations of the agents are correlated through the type of the items, finding the optimal static auction is challenging and could be computationally intractable~\citep{PapadimitriouP11}. \cite{CremerM88} proposed a mechanism that can extract the whole surplus if the valuations are correlated, however, their mechanism is not practical and does not satisfies the desirable ex-post individual rationality property; also see Section~\ref{sec:transient}.}

For the upper-bound, we consider the optimal $T$-round mechanism that knows the type of the items.
\begin{lemma}[Upper-bound] \label{lem:upper_bound}
Let $\M_T^s$ be the optimal $T$-round mechanism that knows the type of the item, $s$. Similarly, $\M_1^s$ corresponds to the optimal (static) mechanism when $T=1$. Then, the revenue of $\M_T^s$, denoted by $\Rev(\M_T^s)$, is bounded by $T\times \Rev(\M_1^s)$. Furthermore, if mechanism $\M_1^s$ is ex-post incentive compatible, then, $\M_T^s$ can be implemented by repeating mechanism $\M_1^s$ at each step $t=1,\cdots,T$.
\end{lemma}

We prove the first claim in the appendix using a reduction argument that reduces a mechanism in a $T$-round setting to a mechanism in a single-round.
%
If mechanism $\M_1^s$ is ex-post incentive compatible, the leakage of information from one round to another does not change the strategy of the bidders.
Recall that for private value settings, ex-post incentive compatibility  implies that truthfulness is a (weakly) dominant strategy for each agent for any realizations of other agents' valuations --- the second price auction with reserve satisfies this property.

Through out this paper, we make the following standard regularity assumption (cf.~\cite{Myerson81}).
%
%
\begin{assumption}[Regularity] \label{ass:regular}
Distribution $F_s$, $s\in\{L,H\}$, with density $f_s$, is regular, i.e., c.d.f. $F_s(v)$ and $\left(v-{1-F_s(v)\over f_s(v)}\right)$ are strictly increasing in $v$ over the support of $F_s$.
\end{assumption}
%
%
Examples of regular distributions include many common distributions such as the uniform, Gaussian, log-normal, etc.

If $s$ is known and $F_s$ is regular, then $\M_1^s$, the optimal mechanism for $T=1$, is the second price auction with reserve price $\rs$ that is the unique solution of 
\begin{equation} \label{eq:opt_reserve}
r-{1-F_s(r)\over f_s(r)}=0\, .
\end{equation}
Therefore, by Lemma~\ref{lem:upper_bound}, we obtain the following.


\begin{theorem}[No ``dynamic" improvement with single type] \label{thm:no_improve_n_equal_1}
If the valuations are drawn i.i.d. from a regular distribution (e.g., $s$ is known and $F_s$ is regular), the optimal mechanism is the second price auction with a constant reserve that is the solution of {\em Eq.~(\ref{eq:opt_reserve})} and there is no benefit from having dynamic reserve prices.
\end{theorem}

The theorem above is similar to the previous results in the literature for settings with a single buyer~\citep{Stokey79,Salant89,HartT88,AcuistiV05} and generalizes their insights to auction environments with multiple buyers.

In the next section, we preset a simple mechanism that exploits correlations between valuations (via types $s$) and the competition among bidders to extract higher revenue than the static mechanism and in fact, for a broad class of distributions of valuations, obtains revenue close to the upper-bound benchmark. Further, the mechanism is approximately incentive compatible.

\section{The Threshold Mechanism} \label{sec:threshold}
In this section, we present the class of threshold mechanisms.

\medskip
A \emph{\bf threshold mechanism} is defined by three parameters and is denoted by $\M(\rho,r_L,r_H)$ where $r_L$ is the initial reserve price. The reserve stays $r_L$ until any of the agents bid above $\rho$, then for all subsequent rounds, the reserve price will increase to $r_H$. If there are no bids above $\rho$, the reserve stays $r_L$ until the end.
\medskip

As we demonstrate in the following, this class of mechanisms (and a generalization of it, presented in Section~\ref{sec:general}, include good candidates for boosting revenue if the modes of $F_L$ and $F_H$ are sufficiently well separated. The idea is to choose $\rho$ such that the valuation of an agent is unlikely to be above $\rho$ if $s=L$, whereas, a valuation exceeding $\rho$ is quite likely if $s=H$. Moreover, as we establish, truthful bidding forms an approximate equilibrium in this case, so for almost all realizations, the mechanism does correctly infer $s$.

%

To convey the intuition behind our incentive compatiblity results, we start with the following (warm up) proposition.
\begin{proposition}
\label{prop:truthful_equilibrium}
Suppose that $F_{L}$ is supported on $[0,\bar{L})$ and $F_H$ is supported on $[\underline{H},\infty)$ for $\bar{L} <\underline{H}$ and $\alpha_i=1$ for at least 2 agents. Consider any $r_L < r_H$. Then, $\M(\rho,r_L,r_H)$, for any $\rho \in (\bar{L},\underline{H})$, is incentive compatible.
\end{proposition}
\begin{proof}{Proof:} 
Consider agent $i$ and round $t$ where $i$ participates (i.e., $x_{it}=1$).
First observe that since bidding truthfully is a (weakly) dominant strategy in the second-price auction, truthfulness is a myopic best-response in our setting.

If $s=L$, then bidding truthfully will not increase the reserve in the future rounds and truthfulness is a (weakly) dominant strategy.
Now suppose $s=H$. If $r=r_H$ in round $t$, then again truthful bidding is a best response, since the reserve will continue to be $r_H$ for the remaining rounds. 
 On the other hand, if $r=r_L$ (and $s=H$), at least one other agent will participate in the auction at time $t$. At the equilibrium, the other agent will bid truthfully, hence above $\rho$, and the reserve will be $r_H$ for the remaining rounds in any case. So bidding truthfully is a best response, since it is myopically a best response.
\end{proof}

We now show that the threshold mechanism is approximately incentive compatible when the support of the distributions overlap and the distribution of the low type is bounded.
We also provide an example that shows a significant boost in the revenue.

\begin{theorem} \label{thm:approx_Nash}
Let $F_L$ be supported on $[0,\bar{L}]$, $\bar{L}<\infty$. 
Let $\rs$ be the solution of $r-{1-F_s(r)\over f_s(r)}=0$.
Consider any positive $\eps < r_H^\star - \rL$. Let $\alpha = \min_i \alpha_i > 0$. Define
\begin{align*}
n_0&\equiv   1+\frac{1.59 \log(2(r_H^\star-r_L^\star)/\eps)}{(1-F_H(\rho))} < \infty \, ,\\
T_0 &\equiv \frac{2(r_H^\star -r_L^\star)}{\eps} \left \lceil \frac{n_0-1}{(n-1) \alpha} \right \rceil \, .
\end{align*}
Consider $\rho \geq \bar{L}$ such that $(1-F_H(\rho))>0$. Then, Mechanism $\M(\rho,\rL,\rH)$ is $\epsilon$-incentive-compatible for all $n\geq n_0$ and $T \geq T_0$. In addition, the expected revenue for each $s \in \{L,H\}$ is at least $(\Rev(\M_1^s) - \eps) T$,
where $\M_1^s$ is the optimal single-round mechanism that knows $s$ in advance and can be obtained using a constant reserve of $\rs$.
%
\end{theorem}
%
Thus, using mechanism $\M(\rho,\rL,\rH)$ in this setting, truthful bidding is an approximate equilibrium and the revenue is very close to the benchmark. 

Defining $\delta = \eps/(\rH-\rL)$, an appealing feature of Theorem~\ref{thm:approx_Nash} is that the lower bound on number of bidders, $n_0$, grows only as $O(\log(1/\delta))$. 
On the other hand, the lower bound on number of rounds $T_0$ grows as $O(\log(1/\delta)/\delta)$ for $n \alpha = \Omega(1)$. This is somewhat larger than $n_0$ for small $\eps$ but this is not a major concern since the number of identical (or very similar) impressions is often large in online advertising settings.
The below example demonstrates a numerical illustration of Theorem~\ref{thm:approx_Nash}.

\begin{figure} \label{fig:example_visualization}
\centering
  \includegraphics[scale=.3]{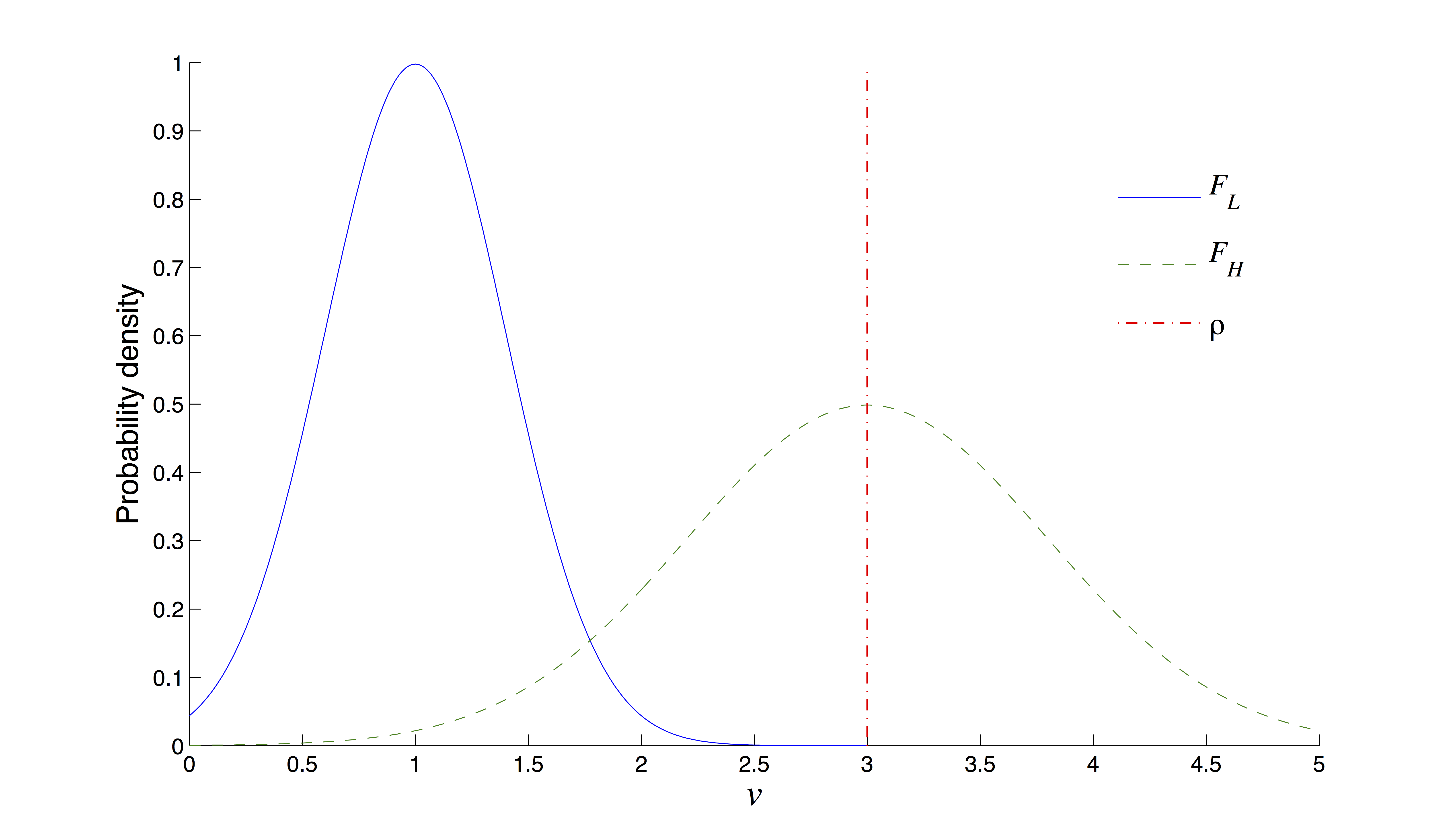} 
\caption{Illustration of distributions in Example \ref{ex:overlapping_gaussians}. $F_L$ is the normal distribution, with mean $1$ and standard deviation $0.4$, truncated to interval $[0,3]$,  and $F_H = \N(3,0.8^2)$, i.e., normal with mean $3$ and standard deviation $0.8$. We use $\rho=3$.}
\end{figure}

\vspace{-10pt}
\begin{example} \label{ex:overlapping_gaussians}
{Suppose $F_L$ is the normal distribution, with mean $1$ and standard deviation $0.4$, truncated to interval $[0,3]$, and $F_H = \N(3,0.8^2)$, i.e., normal with mean $3$ and standard deviation $0.8$. These distributions are shown in Figure \ref{fig:example_visualization}.
Also, let 
$n\alpha = 1$, $\alpha_i=\alpha$ for all $i$.
Note that each agent participates in $\alpha T = T/n$ rounds on average, making it non-trivial to have dynamic reserves without losing incentive compatibility if $T/n$ is much larger than $1$.
We have $\rL \approx 0.796$ and $\rH\approx 2.318$. 
Using $\rho=3$ gives $n_0 \approx 19.52$. Using $n=20$, we obtain $T_0 \approx 6800$ for $\eps=0.009$ so we consider $T=T_0=6800$ in our simulations. 
The (optimal) static second price auction obtains average-revenue per-round equal to $0.755$ using (constant) reserve price $1.05$. 
Mechanism $\M(\rho,\rL,\rH)$ yields per-round revenue of  $0.935$  (Theorem \ref{thm:approx_Nash} guarantees that the loss relative to the optimal revenue is $\eps=.009$ at most per round) improving more than $23\%$ over the static mechanism. The per-round revenue of the optimal mechanism that knows the type of the impressions is equal to $0.938$. The 95\% confidence error in estimating the revenues is less than $0.007$. The average welfare of a buyer per round of participation, averaged over $s$ and $v_i$ is found to be $0.335$ (with 95\% confidence error 0.002).
Note that given this implies that the threshold mechanism is $\eps = 0.009$-incentive compatible. 
%
}\end{example}


\vspace{-10pt}
\section{Dynamic Incentive Compatibility} \label{sec:DIC} 

In this section, we show that, with high probability, no agent has a large incentive to deviate from the truthful strategy in later rounds after acquiring new information. We also relax the requirement that $F_L$ needs to have bounded support.





\begin{theorem} \label{thm:approx_BNE}
Recall \eqref{eq:opt_reserve}. Let $\lambda = 1-F_L(\rho)$.
Consider any $\eps <r_H^\star - r_L^\star$ and let $\delta = \eps/(r_H^\star - r_L^\star)$.
Then, Mechanism $\M(\rho, \rL, \rH)$ is $(\delta, \eps)$-dynamic incentive compatible for all $n\in [n_1,n_2]$ and $T \geq T_1$ where
\begin{align*}
n_1&\equiv   1+\frac{3.18 \log(2/\delta)}{(1-F_H(\rho))} < \infty \, ,\\
n_2 &\equiv \delta/\lambda \, ,\\
T_1 &\equiv \frac{4}{\delta} \left \lceil \frac{n_1-1}{(n-1) \alpha} \right \rceil \, .
\end{align*}
Further, the expected revenue of the mechanism is additively within $\eps T$ of the benchmark revenue under truthful bidding.
\end{theorem}

%
Note that this theorem requires an $n$ to be not too large. The assumed upper bound on $n$ can be eliminated in two different ways: In the (immediate) corollary below, we assume a bounded support for $F_L$ leading to $\lambda=0 \Rightarrow n_2 = \infty$. Later in Section \ref{sec:general}, we introduce a generalized threshold mechanism, which then facilitates a result similar to Theorem \ref{thm:approx_BNE} while allowing $n$ to be arbitrarily large (Theorem \ref{thm:approx_BNE_unbounded}).

\begin{corollary}[Bounded Support] \label{cor:approx_BNE}
Recall \eqref{eq:opt_reserve}. Let $F_L$ be supported on $[0,\bar{L}]$, $\bar{L}<\infty$. Consider any $\eps <r_H^\star - r_L^\star$ and let $\delta = \eps/(r_H^\star - r_L^\star)$.
Then, Mechanism $\M(\rho, \rL, \rH)$ is $(\delta, \eps)$-dynamic incentive compatible for all $n\geq n_1$ and $T \geq T_1$ where
\begin{align*}
n_1&\equiv   1+\frac{3.18 \log(2/\delta)}{(1-F_H(\rho))} < \infty \, ,\\
T_1 &\equiv \frac{4}{\delta} \left \lceil \frac{n_1-1}{(n-1) \alpha} \right \rceil \, .
\end{align*}
Further, the expected revenue of the mechanism is additively within $\eps T$ of the benchmark revenue under truthful bidding.
\end{corollary}

Comparing with Theorem \ref{thm:approx_Nash}, we see that the cost of the stronger notion of equilibrium here is only a factor $2$ loss in $n_1$ and a further factor $2$ loss in $T_1$.


We prove Theorem~\ref{thm:approx_BNE} in the appendix. 
We state below the main lemma leading to a proof of dynamic incentive compatibility $s=H$.

Let $Q_t$ be the event that the reserve in round $t+1$ is $r_H^\star$ assuming truthful bidding (thus $Q_t$, here, is the event that  a bidder with valuation exceeding $\rho$ participates in one of the first $t$ rounds). Let $Q_t^{\sim i}$ be the event that the reserve in round $t+1$ would have been $r_H^\star$ assuming truthful bidding even if the bids of agent $i$ are removed (thus $Q_t^{\sim i}$, here, is the event that a bidder $j\neq i$ with valuation exceeding $\rho$ participates in one of the first $t$ rounds).
Let
\begin{align}
  t_\delta =  \min  \{t: \exists i \textup{ s.t. } \Pr\big( Q_t^{\sim i} \,\big |s=H\big ) \geq 1- \delta \}
  \label{eq:tdelta_def}
\end{align}
(It turns out that $t_\delta \leq\lceil \frac{n_1-1}{(n-1)\alpha }\rceil$.)
By definition, $Q_{t} \supseteq Q_{t}^{\sim i}$ for all $i$ and all $t$. It follows that $\Pr(Q_{t_\delta}|s=H) \geq 1- \delta$, so, in establishing $(\delta, \eps)$-dynamic incentive compatibility, we can ignore the trajectories under which $Q_{t_\delta}$ does not occur (these trajectories have combined probability bounded above by $\delta$). Under $Q_{t_\delta}$, the reserve in round $t_\delta+1$  (and all later rounds) is already $r_H^\star$, making truthful bidding an exact best response in those rounds. Also, for agents whose valuation is less than $\rho$, truthful bidding is always a best response. For an agent $i$ with valuation exceeding $\rho$, on the equilibrium path, the only time that $i$ may potentially benefit from not being truthful is the first time that $i$ participates (and if the reserve is still $r_L^\star$); once $i$ has bid truthfully once, future bids of $i$ have no impact on the reserve and truthful bidding is a best response.\footnote{In the present setting with mechanism $\M(\rho, \rL, \rH)$, agent $i$ will see a reserve that has already risen to $r_H^\star$ in each subsequent round that $i$ participates in. (Later we will generalize the threshold mechanism in Section~\ref{sec:general}, but it will still be true that if $i$ bids above $\rho$ once, future bids of $i$ will not affect the reserve, hence truthful bidding will be exactly optimal in subsequent rounds.)}  Hence, it suffices to show that under $s=H$ for $t\leq  t_\delta$, if an agent $i$ participates for the first time in round $t$, truthful bidding is (additively) $\alpha(T-t_\delta)\eps$-optimal, assuming that others bid truthfully, even if the reserve is still $r_L^\star$ in round $t$.

\begin{lemma}  \label{lemma:BPE_truthfulness}
Assume that $s=H$, $n\geq  n_1\, $, $\eps < (r_H^\star - r_L^\star)$.
Let $t_\delta$ be defined as in Eq.~\eqref{eq:tdelta_def}.
For any agent with valuation exceeding $\rho$ who participates for the first time in a round $t \leq t_\delta$,
and sees a reserve $r_L^\star$, truthful bidding is (additively) $\alpha(T-t_\delta)\eps$-optimal, assuming that others bid truthfully.
Further, we have
$t_\delta \leq \tau = \left \lceil \frac{(n_1-1)}{(n-1)\alpha }\right \rceil \, $
and $\Pr(Q_{t_\delta}|s=H) \geq 1- \delta$. 
\end{lemma}



\section{The Generalized Threshold Mechanism} \label{sec:general}
We now present a generalization of the threshold mechanism that allows us to significantly weaken the required bound on the right tail of the low type distribution of Theorem~\ref{thm:approx_BNE}.


The \emph{generalized threshold mechanism} is defined by four parameters and is denoted by $\M(\rho,r_L,r_H,k)$ where $r_L$ is the initial reserve price. The reserve stays $r_L$ until $k$ distinct agents bid above $\rho$ (possibly in different rounds). If this occurs then for all subsequent rounds, the reserve price will increase to $r_H$. 

\begin{theorem} \label{thm:approx_BNE_unbounded}
Recall \eqref{eq:opt_reserve}. Let $\lambda = 1-F_L(\rho)$. Assume $\lambda \leq (1-F_H(\rho))/18$.
Fix positive $\eps < r_H^\star - r_L^\star$. Define $\delta=\eps/(r_H^\star-r_L^\star)$.
 Let
\begin{align*}
n_3&\equiv   1+8.48 \log(2/\delta)/(1-F_H(\rho)) < \infty \, ,\\
n_4 &\equiv 0.56 \log (2/\delta)/\lambda \, ,\\
\bar{n} &\equiv \max(n_3,n_4)\, ,\\
T_1 &\equiv \frac{4}{\delta} \left \lceil \frac{n_3-1}{(\min(n, \bar{n})-1) \alpha} \right \rceil   .
\end{align*}
We provide mechanisms that work well for any $n \geq n_3$ and $T \geq T_1$. 
\begin{itemize}
  \item Suppose $n_3< n_4$. For all $n \in [ n_3, n_4]$ and $T \geq T_1 = \frac{4}{\delta} \left \lceil \frac{n_3-1}{(n-1) \alpha} \right \rceil$, the generalized threshold mechanism
  $\M(\rho, \rL, \rH, 2.26 \log(2/\delta))$ is $(\delta, \eps)$-dynamic incentive compatible, and it is additively $\eps T$ close to the revenue benchmark.
  \item For all $n \geq \max(n_3,n_4)$ and $T \geq T_1 = \frac{4}{\delta} \left \lceil \frac{n_3-1}{(\bar{n}-1) \alpha} \right \rceil$, the generalized threshold mechanism $\M(\rho, \rL, \rH, 4 \lambda n)$ is $(\delta, \eps)$-dynamic incentive compatible, and it is additively $\eps T$ close to the revenue benchmark.
\end{itemize}

\end{theorem}
As a remark to ease the burden of notation: note that $T_1 \leq \frac{4}{\delta} \left \lceil \frac{1}{ \alpha} \right \rceil$ for the $n$ values of interest, i.e., for $n \geq n_3$. In other words, $\frac{4}{\delta} \left \lceil \frac{1}{ \alpha} \right \rceil$ rounds suffice to obtain our positive results. Also, note that for $n\alpha = \Theta(1)$, we still have $n_3 = O(\log(1/\delta))$ and $T_1 = O(\log(1/\delta)/\delta)$ as was the case for Theorem \ref{thm:approx_Nash}, so our requirements on the number of bidders and number of rounds needed continue to be reasonable.

\section{Discussion} 

\subsubsection*{Transient Valuations} \label{sec:transient}
So far we assumed that the valuations of the agents are constant over time. In this section, we consider the following extension of our model: each time an agent participates, he draws a new valuation from $F_s$ independently with some probability $\beta$ (our original model  corresponds to the case $\beta =0$), and retains his previous valuation with probability $1-\beta$.

We observe that our incentive compatibility result of Theorem~\ref{thm:approx_Nash} (also Corollary \ref{cor:approx_BNE}) holds in this setting because the incentive of the agents to deviate is even smaller and the proof 
work nearly as before for any $\beta \in [0,1]$. 
In addition, Theorem~\ref{thm:no_improve_n_equal_1} holds if we consider mechanisms that are periodic ex-post individually rational~\citep{BergemannV10}; in other words, the utility of any truthful agent at the end of each round $t$, $1\le t\le T$ should be non-negative.

The following example shows that a mechanism that is not ex-post individually rational can obtain a higher revenue by charging the agents a high price in advance: Suppose there is only one agent ($n=1$), the agent participates in all rounds ($\alpha=1$),  and the agent draws a new valuation at each round from the uniform distribution over $[0,1]$ ($\beta=1$). It is not difficult to see that the optimal constant reserve for this setting is equal to ${1\over 2}$ which yields the expected revenue of ${T\over 4}$ since the agent will purchase the item with probability ${1\over 2}$. Now consider a mechanism that offers reserve price ${T-1\over 2}-\varepsilon$ (for an arbitrarily small $\varepsilon$) in the first round and if the agent accepts that price, the mechanism offers the item for free in the future rounds, and if the agent refuses the offer, the mechanism posts a price of $1$ at each round. Observe that the agent will accept the mechanism's offer in the first round and the revenue obtained in this case is equal to ${T-1\over 2}-\varepsilon$. However, this mechanism is not ex-post individually rational. For $\beta\in(0,1)$ the optimal mechanism (that does not satisfy the ex-post IR property) would take the form of contracts followed by sequence of auctions~\citep{KakadeLN13,Battaglini05}.

\subsubsection*{Connection to Mean Field Equilibrium}
We now comment briefly on the connection between our work, and the concept of mean field equilibrium. A number of recent papers study notions of mean field equilibrium, e.g., \cite{IyerJS11,BalseiroBW13} study mean field equilibria in dynamic auctions, and \cite{GummadiKP13} studies mean field equilibrium in multiarmed bandit games. An agent making a mean field assumption assumes that the set of competitors (or cooperators) she faces will be drawn uniformly at random from a large pool of agents with a known distribution of types. In our work, agents' participate in a particular round of a dynamic auction independently at random, but our results do not require $n \rightarrow \infty$ and agents retain their valuation for all rounds in which they participate. 
In Theorem \ref{thm:approx_BNE}, one can have any fixed number of agents exceeding $n_1 = O(\log (1/\eps))$, and participants reason about the posterior distribution of competitors they will face in a round, given the information available to them. This posterior distribution of the valuations of competitors is in general different from the prior distribution of valuations and evolves from one round to the next.

\section{Conclusion} \label{sec:conclusion}

We considered repeated auctions of items, all of the same type, with the auctioneer not knowing the type of the items a-priori. In our model, the issue of incentives is challenging because a bidder typically participates in multiple auctions, and is hence sensitive to changes in future reserve prices based on current bidding behavior. We demonstrated a fairly broad setting in which a simple dynamic reserve second price auction mechanism can lead to substantial improvements in revenue over the best fixed reserve second price auction. In fact, our threshold mechanism is approximately truthful and achieves near optimal revenue in our setting. We demonstrate a numerical illustration of our results with a reasonable choice of model parameters, and show significant improvement in revenue over the static baseline.

For our future work, we would like to investigate the effects of various properties of the (joint) distributions of the valuation of the advertisers (e.g.,  more than two types), the characteristics of learning algorithms (as opposed to simple threshold mechanisms), and the effect of the rate (and manner) in which the valuations of advertisers change over time on the equilibrium and the revenue of the auctioneer.



\paragraph{\bf Acknowledgment}
We would like to thank Brendan Lucier, Mohammad Mahdian, and Mukund Sundararajan for their insightful comments and suggestions. This work was supported in part by Microsoft Research New England. The work of the second author was supported in part by a Google Faculty Research Award.


\bibliographystyle{plainnat}
\bibliography{mechanismdesign}

\appendix

\section{Appendix}

\subsection{Proof of Theorem~\ref{thm:approx_Nash}}
We now prove  Theorem \ref{thm:approx_Nash} by first showing that threshold mechanism $\epsilon$-incentive-compatible. 
First assume $s=L$. Consider any agent $i$ with valuation $v_i \in [0,\bar{L}]$ and assume that other agents are truthful always.
Since $v_i \leq \rho$, it is clear that truthful bidding weakly dominates any other strategy, since this is true myopically, the reserve is unaffected, and the bidding behavior of others is unaffected by the bids of agent $i$.
In this case, the reserve remains $r_L^\star$ and the agents bid truthfully throughout, so there is no loss in revenue.

Now assume $s=H$. In Appendix \ref{app:proofs}, we prove the following lemma.


\begin{lemma}  \label{lemma:k_highbidders_by_tau}
 Assume $s=H$, $n\geq n_0 \equiv  1+C \log (2/\delta)/(1-F_H(\rho))$. Fix an agent $i$.
With probability at least $1-(\delta/2)^{C/1.59}$, irrespective of what agent $i$ does, at least $k=1$ bidder
$j \neq i$ with valuation exceeding $\rho$ will bid in the first $\tau = \left \lceil \frac{n_0-1}{(n-1) \alpha} \right \rceil$ rounds.
For $k \leq C \log(2/\delta)/3.18$, at least $k$ bidders different from $i$ with valuation exceeding $\rho$ will bid in the first $\tau = \left \lceil \frac{n_0-1}{(n-1) \alpha} \right \rceil$ rounds with probability at least $1-(\delta/2)^{C/4.24}$.
\end{lemma}

Here we have used $1/(1-e^{-1}) < 1.59$.
We now show how the lemma, for $k=1$, implies the results.

\medskip
\begin{proof}{Proof of Theorem \ref{thm:approx_Nash}.}
  In the first $\tau \leq \frac{\eps T}{2(r_H^\star-r_L^\star)}$ rounds, agent $i$ being truthful can cause the reserve to rise though it wouldn't otherwise have risen, leading to a loss of at most $(r_H^\star-r_L^\star)\alpha_i \tau \leq  \frac{\alpha_i \eps T}{2}$ in expected utility for the agent. If the reserve would not have risen in the first $\tau$ rounds but agent $i$ caused it to rise, this can lead to a further loss of up to $(r_H^\star - r_L^\star)$ per round of participation, and such a loss occurs with probability at most $\eps/(2(r_H^\star-r_L^\star))$ from Lemma \ref{lemma:k_highbidders_by_tau}, leading to a bound of $\frac{\alpha_i \eps T}{2}$ for this loss to the agent. Combining yields the overall bound of $\alpha_i \eps T$ on the loss incurred by agent $i$ by being truthful relative to any other strategy for $s=H$.

Finally we bound the loss in revenue from using this mechanism if $s=H$.
Recall that the optimal auction if the auctioneer knows $s=H$ beforehand is to commit and run a second price auction with reserve $r_H^\star$ in all rounds.
Hence, similar to the above, the expected revenue loss to the auctioneer is bounded by $\tau (r_H^\star -r_L^\star) + \frac{\eps}{2(r_H^\star-r_L^\star)} \cdot T(r_H^\star-r_L^\star) \leq \eps T $ if $s=H$.
Since for each possible $s$, the expected loss in revenue is bounded above by $\eps T$, the same bound holds when we take expectation over $s$.
\end{proof}

\section{Proof of Theorem~\ref{thm:approx_BNE}}
%
First assume $s=L$. A simple union bound ensures that that all bidders have a valuation of at most $\rho$ with probability at least $1-\lambda n \geq 1- \delta$, in which case truthful bidding weakly dominates any other strategy. (Since $v_i \leq \rho$, it is clear that truthful bidding weakly dominates any other strategy, since this is true myopically and the bidding behavior of others is unaffected by the bids of agent $i$.)
Hence, the realization is $0$-good with respect to the mechanism with probability at least $1-\delta$. 
Further, we can easily bound the loss in expected revenue relative to the benchmark under truthful bidding: There is no loss with probability $1-\delta$ (the mechanism matches the benchmark mechanism, since the reserve remains $\rL$ throughout) and a loss of at most $(\rH-\rL)T$ (due to the reserve rising to $\rH$) with probability $\delta$. Thus, the loss in expected revenue is bounded by $\delta (\rH-\rL)T = \eps T$ as required.


Now assume $s=H$. For any agent with a valuation less than or equal to $\rho$, truthful bidding again weakly dominates any other strategy, since this is true myopically and the bidding behavior of others is unaffected by the bids of agent $i$. It remains to deal with agents whose valuation exceeds $\rho$,   to establish that truthful bidding is $(\delta,\eps)$-incentive-compatible. In particular, we need to show that with probability at least $1-\delta$, the realization is $\eps$-good with respect to the mechanism, i.e.,  that each such agent $i$ loses no more than $\alpha_i(T-t)\eps$ in expectation on the equilibrium path from bidding truthfully in round $t$, for each $t$ that $i$ participates in. But this follows from Lemma \ref{lemma:BPE_truthfulness}:  all realizations such that $Q_{t_\delta}$ occurs are $\eps$-good, and $\Pr(Q(t_\delta)) \geq 1-\delta$. See the argument after the statement of Theorem \ref{thm:approx_BNE} in Section \ref{thm:approx_BNE} for further details.

It remains to show that the loss in revenue is no more than $\eps T$, assuming truthful bidding under $\eps$-good realizations. Now, using Lemma \ref{lemma:k_highbidders_by_tau} and $Q_t \supseteq Q_t^{\sim i}$, we have
$$
\Pr(Q_\tau) \geq 1- (\delta/2)^2 \geq 1- \delta/2\, .
$$
Under $Q_\tau$, the mechanism matches the benchmark mechanism for rounds after $\tau$ and hence there is no loss in revenue relative to the benchmark, after the first $\tau$ rounds. In any round, the loss due to setting the wrong reserve (under truthful bidding) is bounded by $\rH-\rL$. Under $\bar{Q}_\tau$, the loss can be this large in each of $T$ rounds, in worst case. It follows that the overall loss in revenue is bounded by $(\rH-\rL) (\tau + \Pr(\bar{Q}_\tau)T)$. But by definition, $\tau = \delta T_1/4 \leq T\delta/4$ and $\Pr(\bar{Q}_\tau) \leq \delta/2$, implying that the loss in revenue relative to the benchmark is at most $(\rH-\rL)\delta(3/4)T = (3/4)\eps T\leq \eps T $ as required, using the definition of $\delta$.
\section{Proof of Theorem \ref{thm:approx_BNE_unbounded}}
We start with the first bullet. The proof for $s=H$ follows exactly the same steps as the proof of Theorem \ref{thm:approx_BNE}, except that we make use of the second part of Lemma \ref{lemma:k_highbidders_by_tau} (using $n \geq n_3$) since we are using $k = 2.26 \log(2/\delta) \leq 8.48\log(2/\delta)/3.18$ instead of $k=1$. Consider $s=L$. The probability of $k$ or more bidders with valuation exceeding $\rho$ is $\Pr(\textup{Binomial}(n,\lambda) \geq k)$. Since $n \leq n_4$, we have the mean of the binomial $\mu = n \lambda \leq \mu_0= 0.56 \log(2/\delta)$, in particular, $k \geq 4 \mu_0 \geq 4 \mu$. Now, using a Chernoff bound (on $\textup{Binomial}(n,\lambda_0)$ where $\lambda_0 = \mu_0/n \geq \lambda$ leading to a mean of $\mu_0$; clearly this binomial stochastically dominates the one we care about), we infer that
\begin{align*}
  \Pr(\textup{Binomial}(n,\lambda) \geq k) \, &\leq \Pr(\textup{Binomial}(n,\lambda_0) \geq k)\\
  &\leq \exp \{ - \mu_0 \cdot 3^2/(2+3)\} = \exp \{ -0.56 \log(2/\delta)\cdot 9/5\} \\
  &\leq  \exp \{-1.00\log(2/\delta) \} = \delta/2 \, .
\end{align*}
If all valuations are no more than $\rho$ then such a realization is clearly $0$-good (i.e., incentive compatible in an exact sense) with respective to the mechanism. Hence, we have shown that the probability of the realization being $\eps$-good is at least $1-\delta/2$, implying $(\delta, \eps)$-dynamic incentive compatibility for $s=L$. Further, the loss in expected revenue for $s=L$ is bounded above by $(\delta/2) T (\rH-\rL) = \eps T/2 \leq \eps T$ as required.

Now consider the second bullet. Consider $s=H$. The threshold is $k = 4 \lambda n$. Let $\bar{n} = \max\{n_3,n_4\}$.

\begin{lemma}  \label{lemma:klin_highbidders_by_tau}
  Assume $s=H$, $n\geq \bar{n} \geq n_0 \equiv  1+C \log (2/\delta)/(1-F_H(\rho))$ and $k \leq C \log(2/\delta)n/(3.18\bar{n})$. Fix an agent $i$.
With probability at least $1-(\delta/2)^{C/1.59}$, irrespective of what agent $i$ does, at least $k$ bidders
different from $i$ with valuation exceeding $\rho$ will bid in the first $\tau = \left \lceil \frac{n_0-1}{(\bar{n}-1) \alpha} \right \rceil$ rounds.
 with probability at least $1-(\delta/2)^{C/4.24}$.
\end{lemma}

 To use Lemma \ref{lemma:klin_highbidders_by_tau} we need an upper bound on $\bar{n}$. Note that using $\delta \leq 1$ and $F_H(\rho) \geq 0$, we have $\log(2/\delta)/(1-F_H(\rho))\geq n_0 2$. Hence we have
\begin{align*}
  n_3 \leq \frac{\log(2/\delta)}{1-F_H(\rho)}(8.48 + 1/n_0 2) \leq  \frac{10.0\log(2/\delta)}{1-F_H(\rho)} \leq \frac{\log(2/\delta)}{1.8\lambda}
\end{align*}
using $\lambda \leq (1-F_H(\rho))/18$. It follows that
\begin{align*}
  \bar{n} \leq \frac{\log(2/\delta)}{1.8\lambda} \, .
\end{align*}
With reference to the upper bound on $k$ in Lemma \ref{lemma:klin_highbidders_by_tau}, we deduce that $$8.48 \log(2/\delta)n/(3.18\bar{n}) \geq 2.26 \log(2/\delta)n/\bar{n} \geq 2.26\cdot 1.8 \lambda n \geq 4 \lambda n .$$ Hence, using Lemma \ref{lemma:klin_highbidders_by_tau}, we deduce that under truthful bidding, the reserve rises to $r_H$ within $\tau = \frac{4}{\delta} \left \lceil \frac{n_3-1}{(\bar{n}-1) \alpha} \right \rceil$ with probability at least
$1-(\delta/2)^2$. Following the argument in the proof of Lemma \ref{lemma:BPE_truthfulness} from here, we deduce $(\delta, \eps)$-dynamic incentive compatibility for $s=H$. We also deduce that the loss in expected revenue is small similar to the proof of Theorem \ref{thm:approx_BNE}.

Consider the second bullet and $s=L$. The probability of $k=4 \lambda n$ or more bidders with valuation exceeding $\rho$ is $\Pr(\textup{Binomial}(n,\lambda) \geq 4 \lambda n)$. The mean $\mu = \lambda n \geq 0.56 \log(2/\delta)$ since $n \geq \bar{n} \geq n_3$. We infer using a Chernoff bound that
\begin{align*}
  \Pr(\textup{Binomial}(n,\lambda) \geq 4 \lambda n) \, &\leq \exp \{ - \mu \cdot 3^2/(2+3)\} \leq \exp \{ -0.56 \log(2/\delta)\cdot 9/5\} \\
  &\leq  \exp \{-1.00\log(2/\delta) \} = \delta/2 \, .
\end{align*}
We then complete the proof of approximate dynamic incentive compatibility and revenue optimality exactly as we did for the first bullet with $s=L$.
%

\section{Proofs of Lemmas}
\label{app:proofs}

\begin{proof}{Proof of Lemma~\ref{lem:upper_bound}.}
To prove the first part of the claim, we construct mechanism $\tilde{\M}$ that obtains, in expectation, revenue equal to $\Rev(\M_T^s)/T$. Since by definition $\Rev(\M_1^s)$ is the optimal revenue that can be obtained when $T=1$, we conclude that $\Rev(\M_T^s)\le T \times \Rev(\M_1^s)$.

We construct mechanism $\tilde{\M}$ as follows: Let $\cB\subseteq\{1,\cdots,n\}$ be the set of agents who participate in the one-round auction. Note that each agent $i$ knows his own $x_{it}$ but not $x_{jt}$ for any other agent $j\neq i$. For all agents $j\notin\cB$, draw a (hypothetical) valuation i.i.d. from the distribution of valuations $F_s$. Now consider the probability space generated by simulating mechanism $\M_T^s$ over a $T$ round auction by sampling $X_{jt}$'s in each round and emulating the (optimal) bidding strategy of the agents under $\M_T^s$.

Consider the distribution ${\cal D}_{\cB}$ of $(q_{\cB}, p_{\cB})$ in rounds where the set of agents who participate is exactly $\cB$, in this probability space. More precisely, we are considering not a single simulation, but the probability space of possible simulation trajectories. For each trajectory $\omega$, the pair $(q_{\cB}, p_{\cB})$ for each round in which agents $\cB$ participate contributes a weight in $H$ proportional to the probability of trajectory $\omega$.

To determine the payments under $\tilde{\M}$, draw $(q_\cB,p_\cB)$ uniformly from distribution ${\cal D}_{\cB}$. The mechanism $\tilde{\M}$ charges the agents in $\cB$ these amounts $p_{\cB}$ and allocate the items according to $q_{\cB}$.

We argue that the mechanism $\tilde{\M}$ is truthful: It is not hard to see that the ex interim expected utility of an agent $i$ from participating in $\tilde{\M}$ with bid $b_i$ when others bid truthfully, is exactly $1/T$ times the ex interim expected utility of participating in $\M_T^s$ and following his equilibrium strategy for valuation $b_i$ there if others follow their equilibrium strategies. Recall that each agent $i$ knows his own $x_{it}$ but not $x_{jt}$ for any other agent $j\neq i$.  It follows that truthful bidding is an equilibrium in mechanism $\tilde{\M}$. Further, under truthful bidding, it is not hard to see that the expected revenue of mechanism $\tilde{\M}$ is $\Rev(\M_T^s)/T$, as claimed. Note that when $\alpha_i=1$, $1\le i\le n$, the proof would be simplified and could be argued using the revelation principle~\cite{Myerson86}.

We now prove the second part of the claim. Note that if $\M_1^s$ is ex-post incentive compatible, the leakage of information from one round to another does not change the strategy of the bidders. Therefore, repeating mechanism $\M_1^s$ obtain revenue $T\times \M_1^s$ which is the upper-bound revenue.
%
\end{proof}


\begin{proof}{Proof of Lemma \ref{lemma:BPE_truthfulness}.}
Consider any agent $i$. By definition of $t_\delta$, we know that for $t \leq t_\delta$, for all agents $i$ we have
\begin{align}
  \Pr(\bar{Q}_{t-1}^{\sim i})> \delta \, .
  \label{eq:Qt_ub}
\end{align}

Let $\tau = \lceil \frac{n_1-1}{(n-1)\alpha }\rceil$.
Note that $\tau \leq \delta T_1/4  \leq \delta T/4 \Rightarrow \tau \leq \delta (T-\tau)/2$.
It follows from Lemma \ref{lemma:k_highbidders_by_tau} that
\begin{align}
  \Pr(\bar{Q}_{\tau}^{\sim i}) \leq \delta^2/2
  \label{eq:Qtaubar_ub}
\end{align}
In particular, we have $t_\delta < \tau$ and $\Pr(Q_\tau) \geq \Pr(Q_{\tau}^{\sim i}) \geq 1-\delta^2/2 \geq 1- \delta$, yielding the second part of the lemma.

Combining Eqs.~\eqref{eq:Qt_ub} and \eqref{eq:Qtaubar_ub} we obtain
that
\begin{align}
  \Pr(\bar{Q}_{\tau}^{\sim i}|\bar{Q}_{t-1}^{\sim i}) \leq \Pr(\bar{Q}_{\tau}^{\sim i})/\Pr(\bar{Q}_{t-1}^{\sim i}) \leq (\delta^2/2)/\delta = \delta/2 \, .
\end{align}
Hence, agent $i$ who participates for the first time in round $t$ and sees reserve $r_L^\star$,
infers that the reserve will rise to $r_H^\star$ by round $\tau+1$ with probability at least $1-\delta/2$,
due to the bids of other agents. Thus, we can bound the expected cost in future rounds to agent $i$ by causing the reserve
to rise by bidding truthfully:
\begin{itemize}
\item Under $\bar{Q}_{\tau}^{\sim i}$, agent $i$ may lose at most $\alpha_i(T-t)(r_H^\star - r_L^\star)$ in future rounds (in expectation).
 \item Under $Q_{\tau}^{\sim i}$, agent $i$ may lose at most $\alpha_i(\tau-t)(r_H^\star - r_L^\star) \leq \alpha_i\tau (r_H^\star - r_L^\star)$ in rounds only up to round $\tau$.
\end{itemize}
Thus, the overall future cost of bidding truthfully is bounded by
\begin{align*}
 &\phantom{=} \Pr(Q_{\tau}^{\sim i}) \alpha_i\tau (r_H^\star - r_L^\star) + \Pr(\bar{Q}_{\tau}^{\sim i})\alpha_i(T-t)(r_H^\star - r_L^\star)\\
  &\leq 1 \cdot \alpha_i (r_H^\star - r_L^\star) \delta (T-\tau)/2 + (\delta/2)\cdot \alpha_i(T-t)(r_H^\star - r_L^\star)\\
  &\leq \delta\alpha_i(T-t)(r_H^\star - r_L^\star)\\
  &= \eps   \alpha_i(T-t)
\end{align*}
as required. Here we used $\tau \leq \delta (T-\tau)/2$ and $t \leq \tau$ from the discussion above.
\end{proof}

\begin{proof}{Proof of Lemma \ref{lemma:k_highbidders_by_tau}.}
Let $Q_{\tau}^{\sim i}(k)$ denote the event of interest, and $Q_{\tau}^{\sim i}$ be the event for $k=1$.
For each agent $j\neq i$, agent $j$ participates in \emph{some}
round $t'$ for $t'\leq \tau$ with probability
$1-(1-\alpha_j)^\tau \geq 1-(1-\alpha)^\tau$. Independently,
agent $j$ has a valuation exceeding $\rho$ with probability $1-F_H(\rho)$.
Hence, we have $v_j \geq \rho$ \emph{and} agent $j$ enters a bids before the end of round $\tau$, with probability
at least $(1-(1-\alpha)^\tau)(1-F_H(\rho))$, and this occurs independently for $j \neq i$.
Note that $(1-x)^{1/x}\leq e^{-1}$ for $x\in (0,1)$ since $(1-x)^{1/x}$ is monotone decreasing in $x$.
Using this bound, we have
\begin{align}
  1-(1-\alpha)^{\tau} &\geq 1 - \exp(-\alpha\tau) \nonumber\\
  &\geq 1- \exp(-(n_0-1)/(n-1)) \nonumber\\
  &\geq \frac{n_0-1}{n-1}(1-\exp(-1)) \nonumber\\
  &\geq \frac{n_0-1}{1.59(n-1)}\, ,
  \label{eq:alpha_term_lb}
\end{align}
where we also used the definition of $\tau$ and convexity of $f(x)=\exp(-2kx)$.

It follows that
\begin{align}
  \Pr(\bar{Q}_{\tau}^{\sim i})&\leq  \Pr\big(\, \textup{Binomial}\big(n-1, [1-(1-\alpha)^\tau][1-F_H(\rho)]\,\big ) \, = \, 0 \, \big )\nonumber\\
  &= \big(\,1-[1-(1-\alpha)^\tau][1-F_H(\rho)]\, \big)^{n-1} \, .
  \label{eq:Qt_ub}
\end{align}

Hence,
\begin{align*}
  \Pr(\bar{Q}_{\tau}^{\sim i}) &\leq  \exp\big \{ [1-(1-\alpha)^\tau][1-F_H(\rho)] (n-1)\big \} \\
  &\leq  \exp\big \{ -(n_0-1)(1-F_H(\rho))/1.59 \big \} \\
  &\leq \exp\{-(C/1.59) \log(2/\delta) \} =(\delta/2)^{C/1.59}\, ,
\end{align*}
using $n_0\geq 1 + C\log(2/\delta)/(1-F_H(\rho))$ and Eq.~\eqref{eq:alpha_term_lb}.

Similarly,
\begin{align}
  \Pr(\bar{Q}_{\tau}^{\sim i}(k))&\leq  \Pr\big(\, \textup{Binomial}\big(n-1, [1-(1-\alpha)^\tau][1-F_H(\rho)]\,\big ) \, < \, k \, \big )\nonumber\\
  &= \big(\,1-[1-(1-\alpha)^\tau][1-F_H(\rho)]\, \big)^{n-1} \, .
  \label{eq:Qtk_ub}
\end{align}
The mean of the binomial is
\begin{align}
\nonumber  \mu = (n-1) [1-(1-\alpha)^\tau][1-F_H(\rho)] \geq (n_0-1)[1-F_H(\rho)]/1.59 = C \log(2/\delta)/1.59
\end{align}
using Eq.~\eqref{eq:alpha_term_lb}.
It follows using a Chernoff bound and $k  \leq C \log(2/\delta)/3.18 \leq \mu/2$ that
\begin{align}
\nonumber  \Pr(\bar{Q}_{\tau}^{\sim i}(k)) \leq \exp\{-\mu(1-1/2)^2/2 \}= \exp\{-C\log(2/\delta)/4.24 \} = (\delta/2)^{C/4.24}.
\end{align}
\end{proof}

\begin{proof}{Proof of Lemma \ref{lemma:klin_highbidders_by_tau}.}
The proof is very similar to the proof of Lemma \ref{lemma:k_highbidders_by_tau}.

Let $Q_{\tau}^{\sim i}(k)$ denote the event of interest.
For each agent $j\neq i$, agent $j$ participates in \emph{some}
round $t'$ for $t'\leq \tau$ with probability
$1-(1-\alpha_j)^\tau \geq 1-(1-\alpha)^\tau$. Proceeding as before, we have
\begin{align}
  1-(1-\alpha)^{\tau} &\geq \frac{n_0-1}{1.59(\bar{n}-1)}\, .
  \label{eq:alpha_term_lb_lin}
\end{align}

We have Eq.~\eqref{eq:Qtk_ub} for the probability of $\bar{Q}_{\tau}^{\sim i}(k)$ as before.
The mean of the binomial is
\begin{align}
  \mu &= (n-1) [1-(1-\alpha)^\tau][1-F_H(\rho)] \geq (n_0-1)[1-F_H(\rho)](n-1)/(1.59(\bar{n}-1)) \nonumber \\
  &= C \log(2/\delta)(n-1)/(1.59(\bar{n}-1)) \geq C \log(2/\delta)n/(1.59\bar{n})
\end{align}
using Eq.~\eqref{eq:alpha_term_lb_lin} and $n \geq \bar{n}$.

It follows using a Chernoff bound and $k  \leq C \log(2/\delta)n/(3.18\bar{n}) \leq \mu/2$ that
\begin{align}
  \Pr(\bar{Q}_{\tau}^{\sim i}(k)) \leq \exp\{-\mu(1-1/2)^2/2 \}= \exp\{-C\log(2/\delta)/4.24 \} = (\delta/2)^{C/4.24} \, ,
\end{align}
using $n \geq \bar{n}$.
\end{proof}

\end{document}